\documentclass[shortnote,twocolumn]{jpsj3}

\usepackage[T1]{fontenc}
\usepackage{type1cm}
\usepackage{newtxtext,newtxmath}

\usepackage{graphicx}
\usepackage{bm}
\usepackage{physics}
\usepackage{tikz}
\usepackage{ulem}

\usepackage{amsthm}
\usepackage{mathrsfs}
\usepackage{mathtools}

\newtheorem{thm}{Theorem}
\newtheorem{prop}{Proposition}
\newtheorem{cor}{Corollary}
\newtheorem{lemma}{Lemma}

\title{Circuit Optimization for Universality Transformation}

\author{\scalebox{0.96}{Yasuaki Nakayama$^1$\thanks{yasuaki.nakayama@ntt.com}, Yuki Takeuchi$^2$\thanks{Takeuchi.Yuki@bk.MitsubishiElectric.co.jp}, and Seiseki Akibue$^3$\thanks{seiseki.akibue@ntt.com}}}
\inst{$^{1,2,3}$ NTT Communication Science Laboratories, NTT Inc., 3-1 Morinosato Wakamiya, Atsugi, Kanagawa 243-0198, Japan\\
$^2$ Information Technology R\&D Center, Mitsubishi Electric Corporation, 5-1-1 Ofuna, Kamakura, Kanagawa 247-8501, Japan\\
$^{2,3}$ NTT Research Center for Theoretical Quantum Information, NTT Inc., 3-1 Morinosato Wakamiya,
Atsugi, Kanagawa 243-0198, Japan}

\abst{It is known that a computationally universal gate set $\{H,CCZ\}$ can be transformed to a strictly universal one $\{H, \Lambda(S)\}$ using one maximally imaginary state $\ket{+i}$ and non-imaginary ancillary qubits. We succeed this transformation with a shorter circuit that eliminates non-imaginary ancillary qubits. We further extend this to the continuous gate-set setting, showing that any multi-qubit unitary can be exactly generated by real single-qubit unitary gates, $CCZ$ gates and $\ket{0}\ket{+i}$.}

\begin{document}

\maketitle

There are two different notions of universality in quantum gate sets: strict and computational universalities~\cite{Aharonov03}. The former is defined as a set of unitary operators that generates a dense subgroup in $SU(2^m)$ for any $m\ge m_0$, where $m_0$ is a constant. Therefore, it is sufficient to generate any unitary matrix. The latter is sufficient to efficiently generate any output probability distribution of quantum circuits, but not to generate all unitary matrices. An example of the strictly universal gate set is $\{H,\Lambda(S)\}$, and that of the computationally universal gate set is $\{H,CCZ\}$. It is important to study how to transform $\{H,CCZ\}$ to $\{H,\Lambda(S)\}$ to understand the gap between the two universalities. Here, we write the Hadamard gate, $S$ gate, Pauli-$Z$ gate, controlled-$S$ gate and controlled-controlled-$Z$ gate:
\begin{equation}
    \begin{aligned}
        &H = \frac{1}{\sqrt{2}}\begin{bmatrix}
        1 & 1 \\
        1 & -1 \\
    \end{bmatrix}, \, S= \begin{bmatrix}
        1 & 0 \\
        0 & i \\
    \end{bmatrix}, \, Z= \begin{bmatrix}
        1 & 0 \\
        0 & -1 \\
    \end{bmatrix},\\
    &\Lambda(S) = \ket{0}\bra{0} \otimes I + \ket{1}\bra{1} \otimes S \, ,\\
    &CCZ = (I^{\otimes 2} - \ket{11}\bra{11}) \otimes I + \ket{11}\bra{11} \otimes Z \, ,
    \end{aligned}
\end{equation}
where $I$ is the two-dimensional identity gate.
Note that $\{H,CCZ\}$ consists of real orthogonal matrices, and hence it cannot generate matrices with imaginary numbers in the computational basis. According to the previous work \cite{Takeuchi24}, $\{H,CCZ\}$ can be transformed to $\{H,\Lambda(S)\}$ by using a maximally imaginary state $\ket{+i} = (\ket{0} + i \ket{1})/\sqrt{2}$, which is an eigenstate of the Pauli-$Y$ operator, together with non-imaginary ancillary qubits $\ket{0}$. In that work, the $S$ gate is generated by $\{H,CCZ\}$, a non-imaginary ancillary qubit $\ket{0}$ and a maximally imaginary state $\ket{+i}$. $\Lambda(S)$ is generated by $\{H,CCZ\}$, $\ket{0}$ and the $S$ gate. See the concrete circuits in Fig. \ref{fig:circuit5} (a), (b) and (c). We study the circuit optimization of the universality transformation. We show that we can generate $\Lambda(S)$ gate by using $\{H,CCZ\}$ and $\ket{+i}$ without non-imaginary ancillary qubits, and succeed in reducing the number of quantum gates. The specific circuit is shown in Fig. \ref{fig:circuit5} (d).

In order to study the circuit optimization of the $S$ gate, we first consider the construction of the rotation operator around $z$ axis $R_{z} (\theta)$ since the $S$ gate is a special case as $\theta = \pi /2$.
Here, we write the rotation matrices around $x$, $y$ and $z$ axes:
\begin{equation}
    \begin{aligned}
        &R_{x} (\theta) =
        \begin{bmatrix}
        \cos \frac{\theta}{2} & -i \sin \frac{\theta}{2}\\
        -i \sin \frac{\theta}{2} & \cos \frac{\theta}{2}\\
        \end{bmatrix} \, ,
        R_{y} (\theta) =
        \begin{bmatrix}
        \cos \frac{\theta}{2} & - \sin \frac{\theta}{2}\\
        \sin \frac{\theta}{2} & \cos \frac{\theta}{2}\\
        \end{bmatrix} \, , \\
        &R_{z} (\theta) =
        \begin{bmatrix}
        e^{-i \theta/2} & 0\\
        0 & e^{i \theta/2}\\
        \end{bmatrix} \, .
    \end{aligned}
\end{equation}

We can show that $R_{z} (\theta)$, which has imaginary numbers in general, can be generated by a real non-local matrix $\Lambda(R_{y}(-2 \theta))$ and $\ket{+i}$. Note that we do not need any non-imaginary ancillary qubit.
\begin{thm}
    The operator $e^{i \theta /2} R_{z}(\theta)$ can be constructed by using a real orthogonal matrix and the maximally imaginary state $\ket{+i}$ without using any non-imaginary ancillary qubit. The concrete construction is as follows:
    \begin{equation}
        U \ket{+i} \ket{\psi} = \ket{+i} \otimes e^{i \theta /2} R_{z} (\theta) \ket{\psi}
    \end{equation}
    where $\ket{\psi}$ is an arbitrary state, and the concrete real orthogonal matrix $U$ is
    \begin{equation}
        U = I \otimes \ket{0} \bra{0} + R_{y} (- 2 \theta) \otimes \ket{1} \bra{1} \, .
    \end{equation}
\end{thm}

\begin{proof}
    By expanding $U \ket{+i} = \ket{+i} \otimes e^{i \theta /2} R_{z} (\theta)$, we obtain
    \begin{equation}
        U (\ket{0} + i \ket{1}) =
        \begin{bmatrix}
        1 & 0\\
        0 & \cos \theta\\
        0 & 0\\
        0 & - \sin \theta\\
        \end{bmatrix} + i
        \begin{bmatrix}
        0 & 0\\
        0 & \sin \theta\\
        1 & 0\\
        0 & \cos \theta\\
        \end{bmatrix} \, .
    \end{equation}
    By comparing the real and imaginary parts of this equation, we obtain $U$ as
    \begin{equation}
        U = \begin{bmatrix}
        1 & 0 & 0 & 0\\
        0 & \cos \theta & 0 & \sin \theta\\
        0 & 0 & 1 & 0\\
        0 & -\sin \theta & 0 & \cos \theta\\
    \end{bmatrix}
    = I \otimes \ket{0} \bra{0} + R_{y} (- 2 \theta) \otimes \ket{1} \bra{1} \, .
    \end{equation}
\end{proof}

Here, we assume the catalytic transformation, in which the input $\ket{+i}$ does not change. We can also consider the situation that $\ket{+i}$ changes to $\ket{-i} = (\ket{0} - i \ket{1})/\sqrt{2}$, which is the only other resource state of the universality transformation \cite{Nakayama26}, by acting $H$ to $\ket{+i}$.
When $\theta$ takes the specific value $\pi/2$, the operator $e^{i \theta /2} R_{z} (\theta)$ becomes the $S$ gate and $R_{y} (-2 \theta)$ becomes $i Y$, and thus the next corollary follows. Note that this is the same result as the previous works \cite{Aliferis06, Takagi17, Gidny17}.

\begin{cor}
    The $S$ gate is generated by the real orthogonal matrix $U = I \otimes \ket{0} \bra{0} + i Y \otimes \ket{1} \bra{1} = (\Lambda(Z)) H (\Lambda(Z)) H$ as $U \ket{+i} \ket{\psi} = \ket{+i} \otimes S \ket{\psi}$ without any non-imaginary ancillary qubit. Here, $H$ acts on the first qubit, and $\Lambda(Z)$ is the controlled-$Z$ ($CZ$) gate.
\end{cor}

By this corollary, we do not need any non-imaginary ancillary qubit.
By increasing the control qubit, we can generate $\Lambda(S)$ with $\{H,CCZ\}$ without using non-imaginary ancillary qubits.

\begin{prop}
    The $\Lambda(S)$ gate is generated by the real orthogonal matrix $U = I \otimes \left( I ^{\otimes 2} - \ket{11} \bra{11} \right) + i Y \otimes \ket{11} \bra{11}$ as $U \ket{+i} \ket{\psi} = \ket{+i} \otimes \Lambda(S) \ket{\psi}$ without any non-imaginary ancillary qubit. The construction that only uses $\{H,CCZ\}$ is possible by $U = (CCZ) H (CCZ) H$ where $H$ acts on the first qubit.
\end{prop}

In the previous work \cite{Takeuchi24}, we needed non-imaginary ancillary qubits to generate $\Lambda(S)$. In this note, we discover a circuit to generate $\Lambda(S)$ by using $\{H,CCZ\}$ and $\ket{+i}$ without any non-imaginary ancillary qubits. We also succeed in reducing the circuit depth, which determines the quantum circuit runtime, compared with the previous work by reducing the number of $CCZ$ gates by at least 75\%.

\begin{figure}
\vspace*{-0.2mm}
\hspace*{-5mm}
\centering
\includegraphics[width=90mm]{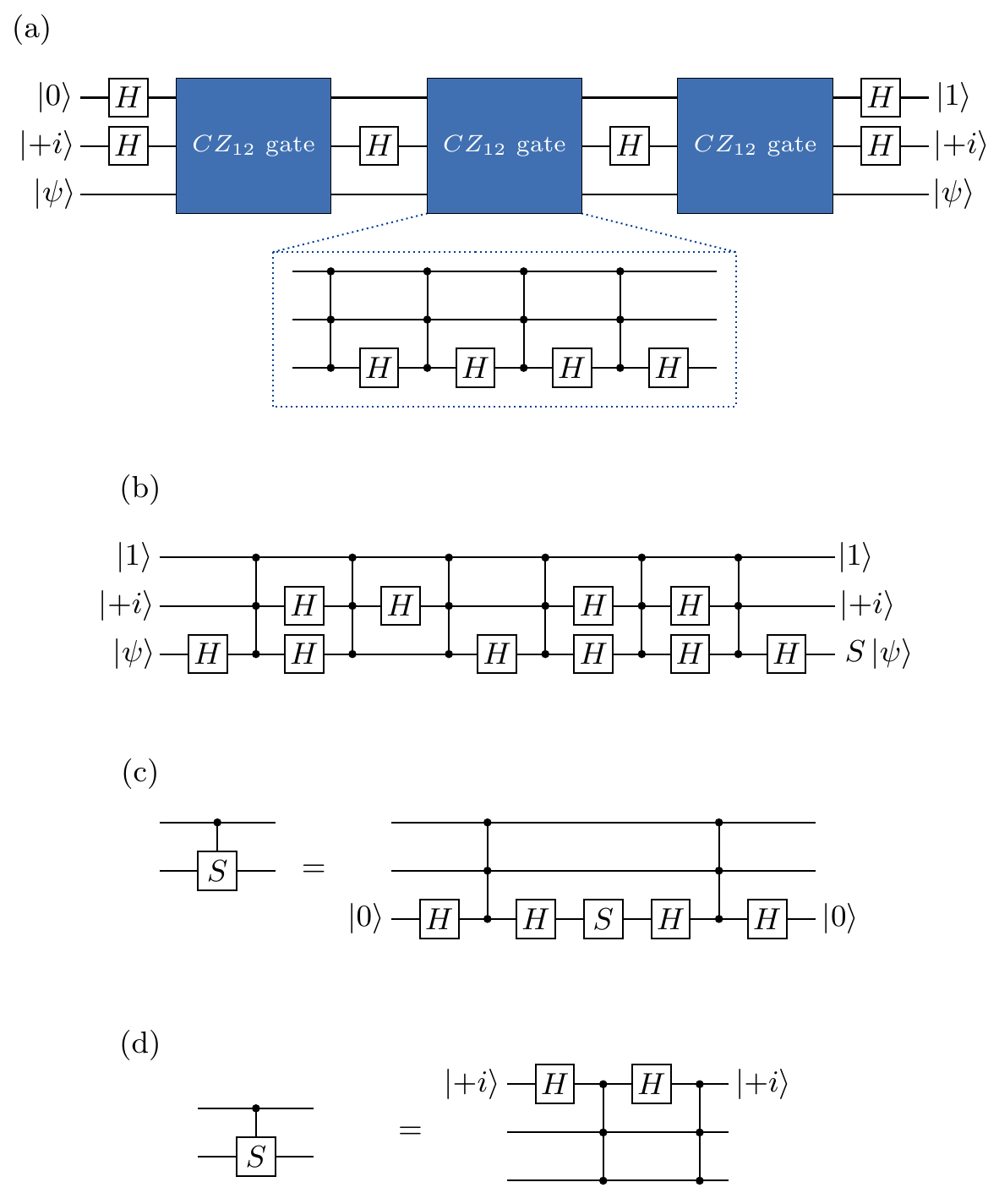}
\caption{The circuits (a), (b) and (c) discovered by the previous work \cite{Takeuchi24} and the new ancilla-free circuit (d) discovered in this note to transform computational universality to strict one. (a) In the previous work \cite{Takeuchi24}, $\ket{1}$ is generated by the construction of $\ket{1}\ket{\phi}$ from $\ket{0}\ket{\phi}$ by using $\{H,CCZ\}$ where $\ket{\phi}$ is an arbitrary two-qubits state. For ease of use, we prepare the input qubits as $\ket{0}\ket{+i}\ket{\psi}$ here. (b) In the previous work \cite{Takeuchi24}, the $S$ gate is generated by $\{ H,CCZ\}$ and $\ket{+i}$ with an ancillary qubit $\ket{1}$, which is generated above. (c) In the previous work \cite{Takeuchi24}, $\Lambda(S)$ is generated by $\{H,CCZ\}$ and the $S$ gate with an ancillary qubit $\ket{0}$. (d) In this note, we show the generation of $\Lambda(S)$ by two $CCZ$ gates without any non-imaginary ancillary qubit. We can reduce the number of $CCZ$ gates by at least 75\% compared to the previous work \cite{Takeuchi24}.}
\label{fig:circuit5}
\end{figure}

From Proposition 1, the $S$ gate can be generated by $\{H,CCZ\}$, one non-imaginary ancillary qubit $\ket{0}$ and $\ket{+i}$ as follows.

\begin{lemma}
    The $S$ gate can be generated by using $\{H,CCZ\}$, one non-imaginary ancillary qubit $\ket{0}$ and $\ket{+i}$.
\end{lemma}

\begin{proof}
    From Proposition 1, we can generate $\Lambda(S)$ by using $\{H,CCZ\}$ and $\ket{+i}$. We can prepare $\ket{1} \ket{+i} \ket{\psi}$ by acting $\{H,CCZ\}$ to $\ket{0} \ket{+i} \ket{\psi}$ by the circuit of Fig. \ref{fig:circuit5} (a). We can generate the $S$ gate by acting $\Lambda(S)$ to $\ket{1}$ such that the first qubit is $\ket{1}$. As a result, we can generate the $S$ gate by using $\{H,CCZ\}$ and $\ket{0}\ket{+i}$.
\end{proof}

\vfill\break

From Fig. \ref{fig:circuit5} (a) and (d), we need $14$ $CCZ$ gates to construct the $S$ gate in Lemma 1. This contributes to the gate reduction from the previous work \cite{Takeuchi24}, in which $18$ $CCZ$ gates are needed to generate the $S$ gate as shown in Fig. \ref{fig:circuit5} (a) and (b).

The cardinality of $\{H, CCZ\}$ is two, and hence the strictly universal gate set transformed from it cannot exactly generate $SU(2^m)$, whose cardinality is uncountable. Any gate set that generates $SU(2^m)$ exactly must contain at least one continuous parameter.
While Theorem~1 provides a method to exactly construct a continuous gate, it requires a non-local unitary matrix with a continuous parameter.
Can we generate $SU(2^m)$ exactly using local unitary matrices that contain only a single continuous parameter?
As an application of Lemma 1, we show that it is indeed possible. In Theorem~2, we show that a gate set consisting of real single-qubit unitaries in $O(2)$, which has a single real parameter, together with the $CCZ$ gate and the resource state $\ket{0}\ket{+i}$ suffices to generate $SU(2^m)$.

\begin{thm}
    If we can use real single-qubit unitary matrices in $O(2)$, the $CCZ$ gates, one non-imaginary ancillary qubit $\ket{0}$ and one qubit of $\ket{+i}$, any unitary matrix in $SU(2^m)$ for any positive number $m$ can be generated.
\end{thm}

\begin{proof}
     We can use real single-qubit unitaries such as $H$ and $R_{y} (\theta)$ for any real number $\theta$. We can also generate $R_{x} (\theta) = S^\dagger R_{y} (\theta) S$ since we can generate the $S$ gate with $\{H,CCZ\}$, $\ket{0}$ and one qubit of $\ket{+i}$ according to Lemma 1. While $\ket{1}$ can be prepared using the method described in Lemma 1, a simpler way is to act the $X$ gate to $\ket{0}$, as the $X$ gate is an element of $O(2)$ and is available here. With the rotation operators around the $x$ axis and $y$ axis, we can generate any unitary operator in $SU(2)$. The \hspace{0.3mm}$CZ$ gate $\Lambda(Z)$ can be generated by acting the $CCZ$ gate to $\ket{1}$ such that the first qubit is $\ket{1}$. It is well known that for any positive number $m$, any unitary operator in $SU(2^m)$ can be constructed by combining $\Lambda(Z)$'s and any unitary operators in $SU(2)$ (see, e.g., \cite{Kitaev02}). As a result, any unitary matrix in $SU(2^m)$ is generated by real single-qubit unitary matrices, $CCZ$ gates, $\ket{0}$ and one qubit of $\ket{+i}$.
\end{proof}

We conclude that real single-qubit unitary gates and $CCZ$ gates can be transformed to any multi-qubit unitary matrix by using one non-imaginary ancillary qubit $\ket{0}$ and one qubit of the maximally imaginary state $\ket{+i}$.
It is worth mentioning that $\ket{0}$ can be made unnecessary by assuming that the $CZ$ gate $\Lambda(Z)$ is available instead of the $CCZ$ gate.


\begin{thebibliography}{9}
\bibitem{Aharonov03} Dorit Aharonov, A simple proof that Toffoli and Hadamard are quantum universal, arXiv:quant-ph/0301040.
\bibitem{Takeuchi24} Yuki Takeuchi, Catalytic Transformation from Computationally Universal to Strictly Universal Measurement-Based Quantum Computation, Physical Review Letters $\mathbf{133}$, 050601 (2024).
\bibitem{Nakayama26} Yasuaki Nakayama, Yuki Takeuchi, and Seiseki Akibue, Uniqueness of imaginarity-assisted transformation from computationally universal to strictly universal quantum computation, arXiv:2603.11812.
\bibitem{Aliferis06} Panos Aliferis, Daniel Gottesman, and John Preskill, Quantum accuracy threshold for concatenated distance-3 codes, Quantum Information and Computation $\mathbf{6}$, 97 (2006).
\bibitem{Takagi17} Ryuji Takagi, Theodore J. Yoder, and Isaac L. Chuang, Error rates and resource overheads of encoded three-qubit gates, Physical Review A $\mathbf{96}$, 042302 (2017).
\bibitem{Gidny17} Craig Gidney and Austin Fowler, A slightly smaller surface code S gate, arXiv:1708.00054.
\bibitem{Kitaev02} Alexei Yu. Kitaev, Alexander H. Shen, and Mikhail N. Vyalyi, \textit{Classical and Quantum Computation} (American Mathematical Society, 2002). 
\end{thebibliography}
\end{document}